\documentclass[a4paper,11pt,english]{article}

\usepackage{babel,amsmath,amssymb,amsthm,graphicx,enumitem,tikz,enumitem,mathtools}
\usepackage[sort&compress,numbers,square,sort]{natbib}

\newtheorem{theorem}{Theorem}
\newtheorem{lemma}[theorem]{Lemma}
\newtheorem{corollary}[theorem]{Corollary}
\newtheorem*{definition*}{Definition}

\let\P\undef
\DeclareMathOperator{\E}{E}
\DeclareMathOperator{\P}{P}
\DeclareMathOperator{\I}{I}
\DeclareMathOperator{\Var}{Var}
\DeclareMathOperator{\esssup}{ess\,sup}
\DeclareMathOperator{\sgn}{sgn}
\DeclareMathOperator{\med}{med}
\DeclareMathOperator{\msr}{\mathbb{S}}
\DeclareMathOperator{\msrp}{\msr_p}
\DeclareMathOperator{\cvar}{CVAR}
\newcommand{\R}{\mathbb{R}}
\newcommand{\geqs}{{\succcurlyeq}}
\newcommand{\leqs}{{\preccurlyeq}}
\newcommand{\PP}{\mathbb{P}}
\newcommand{\QQ}{\mathbb{Q}}

\renewcommand{\Re}{\overline{\mathbb{R}}}
\newcommand{\Z}{\mathcal{Z}}

\renewcommand{\sp}[2]{{#1}\cdot {#2}}
\renewcommand{\L}{\mathrm{L}}
\newcommand{\Lp}{{\L^p}}

\newcommand{\Min}{\wedge}
\renewcommand{\sp}[2]{{\langle #1, #2 \rangle}}
\newcommand{\as}{\text{a.s.}}
\newcommand{\tsum}{{\textstyle\sum}}
\renewcommand{\tilde}{\widetilde}
\renewcommand{\epsilon}{\varepsilon}

\title{Monotone Sharpe ratios and related measures of investment
performance}
\author{Mikhail Zhitlukhin\thanks{Steklov Mathematical Institute of the
Russian Academy of Sciences, 8 Gubkina st., Moscow, Russia. E-mail:
mikhailzh@mi-ras.ru. The research was supported by the Russian Science
Foundation, project 18-71-10097.}}

\date{}

\begin{document}
\maketitle
\begin{abstract}
We introduce a new measure of performance of investment strategies, the monotone Sharpe
ratio. We study its properties, establish a connection with coherent risk
measures, and obtain an efficient representation for using in applications.
\end{abstract}

\section{Introduction}
This paper concerns the problem of evaluation of performance of investment
strategies. By performance, in a broad sense, we mean a numerical quantity
which characterizes how good the return rate of a strategy is, so that an
investor typically wants to find a strategy with high performance.

Apparently, the most well-known performance measure is the Sharpe ratio, the
ratio of the expectation of a future return, adjusted by a risk-free rate or another
benchmark, to its standard deviation. It was introduced by William F. Sharpe in
the 1966 paper \cite{Sharpe66}, a more modern look can be also found in
\cite{Sharpe94}. The Sharpe ratio is based on the Markowitz
mean-variance paradigm \cite{Markowitz52}, which assumes that investors
need to care only about the mean rate of return of assets and the variance of 
the rate of return: then in order to find an investment strategy
with the smallest risk (identified with the variance of return) for a given
desired expected return, one just needs  to find a strategy with the
best Sharpe ratio and  diversify appropriately between this strategy and
the risk-free asset (see a brief review in Section 2 below). Despite its simplicity,
as viewed from today's economic science, the Markowitz portfolio theory was a major
breakthrough in mathematical finance. Even today, more than 65 years later, analysts still routinely compute Sharpe ratios of
investment portfolios and use it, among other tools, to evaluate performance.

In the present paper we look at this theory in a new way, and establish
connections with much more recent developments. The main part of the
material of the paper developed from a well-known observation that variance
is not identical to risk: roughly speaking, one has to distinguish between
``variance above mean'' (which is good) and ``variance below mean'' (which
is bad). In particular, the Sharpe ratio lacks the property of monotonicity,
i.e. there might exist an investment strategy which always yields a return
higher than another strategy, but has a smaller Sharpe ratio. The original
goal of this work was to study a modification of the Sharpe ratio, which
makes it monotone. Some preliminary results were presented in
\cite{Z16,Z17}. It turned out, that the modified Sharpe ratio posses
interesting properties and is tightly connected to the theory of risk
measures. The study of them is the subject of this paper.

The modification of the Sharpe ratio we consider, which we call the monotone
Sharpe ratio, is defined as the maximum of the Sharpe ratios of all probability
distributions that are dominated by the distribution of the return of some given
investment strategy.
In this paper we work only with \emph{ex ante} performance measure, i.e.
assume that probability distributions of returns
are known or can be modeled, and one needs to evaluate their performance; we
leave aside the question how to construct appropriate models and
calibrate them from data. 

The theory we develop focuses on two aspects: on one hand, to place the new
performance measure on a modern theoretical foundation, and, on the other
hand, take into account issues arising in applications, like a possibility
of fast computation and good properties of numerical results. Regarding the
former aspect, we can mention the paper of Cherny and Madan
\cite{ChernyMadan08}, who studied performance measures by an axiomatic
approach. The abstract theory of performance measures they proposed is
tightly related to the theory of convex and coherent risk measures, which
has been a major breakthrough in the mathematical finance in the past two
decades. We show that the monotone Sharpe ratio satisfies those axioms,
which allows to apply results from the risk measures theory to it through
the framework of Cherny and Madan. Also we establish a connection with more
recently developed objects, the so-called buffered probabilities, first
introduced by Rockafellar and Royset in \cite{RockafellarRoyset10} and now
gaining popularity in applications involving optimization under uncertainty.
Roughly speaking, they are ``nice'' alternatives to optimization criteria
involving probabilities of adverse events, and lead to solutions of
optimization problems which have better mathematical properties compared to
those when standard probabilities are used. One of main implications of our
results is that the portfolio selection problem with the monotone Sharpe
ratio is equivalent to minimization of the buffered probability of loss.

Addressing the second aspect mentioned above, our main result here is a representation of the monotone Sharpe ratio as a
solution of some convex optimization problem, which gives
a computationally efficient way to evaluate it. Representations of various
functionals in such a way are well-known in the literature on convex
optimization. For example, in the context of finance, we can mention the famous result of
Rockafellar and Uryasev \cite{RockafellarUryasev00} about the representation
of the conditional value at risk. That paper also provides a good explanation
why such a representation is useful in applications (we also give a brief
account on that below). 

Our representation also turns out to be useful in stochastic control
problems related to maximization of the Sharpe ratio in dynamic trading.
Those problems are known in the literature as examples of stochastic control
problems where the Bellman optimality principle cannot be directly applied.
With our theory, we are able to find the optimal strategies in a shorter and
simpler way, compared to the results previously known in the literature.

Finally, we would like to mention, that in the literature a large number of
performance measures have been studied. See for example papers \cite{LeSourd07,CogneauHubner09a,CogneauHubner09b}
providing more than a hundred examples of them addressing various
aspects of evaluation of quality of investment strategies. We believe that
due to both the theoretical foundation and the convenience for applications, the
monotone Sharpe ratio is a valuable contribution to the field.

The paper is organized as follows. In Section 2 we introduce the monotone
Sharpe ratio and study its basic properties which make it a reasonable
performance measure. There we also prove one
of the central results, the representation as a solution of a
convex optimization problem. In Section 3, we generalize the concept of the
buffered probability and establish a
connection with the monotone Sharpe ratio, as well as show how it can be used
in portfolio selection problems. Section 4 contains applications to dynamic
problems.

\section{The monotone Sharpe ratio}
\subsection{Introduction: Markowitz portfolio optimization and the Sharpe
ratio}
\label{msr intro}
Consider a one-period market model, where an investor wants to distribute
her initial capital between $n+1$ assets: one riskless asset and $n$ risky
assets. Assume that the risky assets yield return $R_i$, $i=1,\ldots,n$, so
that \$1 invested ``today'' in asset $i$ turns into \$$(1+R_i)$
``tomorrow''; the rates of return $R_i$ are random variables with known
distributions, such that $R_i>-1$ with probability 1 (no bankrupts happen).
The rate of return of the riskless asset is constant, $R_0=r>-1$. We always
assume that the probability distributions of $R_i$ are known and given, and,
for example, do not consider the question how to estimate them from past
data. In other words, we always work with \textit{ex ante} performance measures
(see \cite{Sharpe94}).

An investment portfolio of the investor is identified with a vector $x\in
\R^{n+1}$, where $x_i$ is the proportion of the initial capital invested in
asset $i$. In particular, $\sum_i x_i = 1$. Some coordinates $x_i$ may be
negative, which is interpreted as short sales ($i=1,\ldots,n$) or loans
($i=0$). It is easy to see that the total return of the portfolio is $R_x =
\sp xR := \sum_i x_iR_i$.

The Markowitz model prescribes the investor to choose the optimal investment
portfolio in the following way: she should decide what expected return $\E
R_x$ she wants to achieve, and then find the portfolio $x$ which minimizes
the variance of the return $\Var R_x$. This leads to the quadratic
optimization problem:
\begin{equation}
\begin{aligned}
\text{minimize}\quad& \Var R_x\ \text{over}\ x\in \R^{n+1}\\
\text{subject to}\quad& \E R_x = \mu\\
&\tsum_i x_i = 1.
\end{aligned}
\label{markowitz}
\end{equation}
Under mild conditions on the joint distribution of $R_i$, there exists a
unique solution $x^*$, which can be easily written explicitly in terms of
the covariance matrix and the vector of expected returns of $R_i$ (the
formula can be found in any textbook on the subject, see, for example,
Chapter~2.4~in~\cite{Pliska97}).

It turns out that points $(\sigma_{x^*}, \mu_{x^*})$, where $\sigma_{x^*} =
\sqrt{\Var R_{x^*}}$, $\mu_{x_*} = \E R_{x^*}$ correspond to the optimal
portfolios for all possible expected returns $\mu \in [r,\infty)$, lie on
the straight line in the plane $(\sigma, \mu)$, called the efficient
frontier. This is the set of portfolios the investor should choose from --
any portfolio below this line is inferior to some efficient portfolio (i.e.
has the same expected return but larger variance), and there are no
portfolios above the efficient frontier.

The slope of the efficient frontier is equal to the Sharpe ratio of any
efficient portfolio containing a non-zero amount of risky assets (those
portfolios have the same Sharpe ratio). Recall that the Sharpe ratio of return $R$ is
defined as the ratio of the expected return adjusted by the risk-free rate
to its standard deviation
\[
S(R) = \frac{\E (R-r)}{\sqrt{\Var{R}}}.
\]
In particular, to solve problem~\eqref{markowitz}, it is enough to find some
efficient portfolio $\hat x$, and then
any other efficient portfolio can be constructed by a combination of the
riskless portfolio $x_0 =(1,0,\ldots,0)$ and $\hat x$, i.e $x^* = (1-\lambda) x_0 +
\lambda\hat x$, where $\lambda\in[0,+\infty)$ is chosen to satisfy $\E
R_{x^*} = \mu\ge r$.
This is basically the statement of the Mutual Fund Theorem. Thus, the Sharpe
ratio can be considered as a measure of  performance of an investment
portfolio and an investor is interested in finding a portfolio with the
highest performance.  In practice, broad market indices can be
considered as quite close to efficient portfolios.

The main part of the material in this paper grew from the observation that
the Sharpe ratio is not monotone: for two random variables $X,Y$ the
inequality $X\le Y$ a.s. does not imply the same inequality between
their Sharpe ratios, i.e. that $S(X) \le S(Y)$. Here is an example: let $X$
have the normal distribution with mean 1 and variance 1 and $Y= X \Min 1$;
obviously, $S(X)=1$ but one can compute that $S(Y)>1$.
From the point of view of the portfolio selection problem, this fact means
that it is possible to increase the Sharpe ratio by
disposing part of the return (or consuming it). This doesn't agree well with
the common sense interpretation of efficiency. Therefore, one may want to
look for a replacement of the Sharpe ratio, which will not have such a
non-natural property.

In this paper we'll use the following simple idea: if it is possible to
increase the Sharpe ratio by disposing a part of the return, let's define
the new performance measure as the maximum Sharpe ratio that can be achieve
by such a disposal. Namely, define the new functional by
\[
\msr(X) = \sup_{C\ge 0} S(X-C)
\]
where the supremum is over all non-negative random variables $C$ (defined on the same
probability space as $X$), which represent the disposed return.
In the rest of this section, we'll study such functionals and how they
can be used in portfolio selection problems. We'll work in a more general
setting and consider not only the ratio of expected return to standard
deviation of return but also ratios of expected return to deviations in $\L^p$. The
corresponding definitions will be given below.

\subsection{The definition of the monotone Sharpe ratio and its representation}
In this section we'll treat random variables as returns of some investment
strategies, unless other is stated. That is, large values are good, small
values are bad. Without loss of generality, we'll assume that the risk-free
rate is zero, otherwise one can replace a return $X$ with $X-r$, and all the
results will remain valid.

First we give the definition of a deviation measure in $\L^p$,
$p\in[1,\infty)$, which will be used in the denominator of the Sharpe ratio
instead of the standard deviation (the latter one is a particular case for
$p=2$). Everywhere below, $\|\cdot\|_p$ denotes the norm in $\Lp$, i.e.
$\|X\|_p = (\E|X|^p)^{\frac1p}$.

\begin{definition*}
We define the $\Lp$-deviation of a random variable $X\in \L^p$ as
\[
\sigma_p(X) = \min_{c\in \R} \|X - c\|_p.
\]
\end{definition*}
In the particular case $p=2$, as is well-known, $\sigma_2(X)$ is the
standard deviation, and the minimizer is $c^* = \E X$. For $p=1$, the
minimizer $c^*=\med(X)$, the median of the distribution of $X$, so that
$\sigma_1(X)$ is the absolute deviation from the median. It is possible to
use other deviation measures to define the monotone Sharpe ratio, for
example $\|X-\E X\|_p$, but the definition given above seems to be the most
convenient for our purposes.

Observe that $\sigma_p$ obviously satisfies the following properties, which
will be used later: (a) it is sublinear; (b) it is uniformly continuous on $\L^p$; (c)
$\sigma_p(X) = 0$ if and only if $X$ is a constant a.s.; (d) for any
$\sigma$-algebra $\mathcal{G}\subset\mathcal{F}$, where $\mathcal{F}$ is the
original $\sigma$-algebra on the underlying probability space for $X$, we
have $\sigma_p(\E(X\mid\mathcal{G})) \le \sigma_p(X)$; (e) if $X$ and $Y$
have the same distributions, then $\sigma_p(X) = \sigma_p(Y)$.

\begin{definition*}
The \emph{monotone Sharpe ratio in $\Lp$} of a random variable $X\in\Lp$ is defined by
\begin{equation}
\msrp(X) = \sup_{Y\le X} \frac{\E Y}{\sigma_p(Y)},
\label{sp def}
\end{equation}
where the supremum is over all $Y\in \L^p$ such that $Y\le X$ a.s. For $X=0$ a.s. we set by
definition $\msrp(0) = 0$.
\end{definition*}

One can easily see that if $p>1$, then $\msrp(X)$ assumes value in $[0,\infty]$. Indeed, if
$\E X \le 0$, then $\msrp(X) = 0$ as it is possible to take $Y \le X$ with
arbitrarily large $\Lp$-deviation keeping $\E Y$ bounded. On the other hand, if $X\ge0$ a.s. and
$\P(X>0)\neq 0$, then $\msrp(X) = +\infty$ as one can consider $Y_\epsilon =
\epsilon \I(X\ge \epsilon)$ with $\epsilon\to0$ for which $\E
Y_\epsilon/\sigma_p(Y_\epsilon) \to \infty$. 

Thus, the main case of interest will be when $\E X > 0$ and $\P(X<0)\neq0$;
then $0 < \msrp(X) < \infty$. For this case, the following theorem
provides the representation of $\msrp$ as a solution of some
convex optimization problem. 

\begin{theorem}
\label{th msr repr}
Suppose  $X\in \Lp$ and $\E(X) > 0$, $\P(X<0)\neq0$. Then the following
representations of the monotone Sharpe ratio are valid.

\noindent
1) For $p\in(1,\infty)$ with    $q$ such that $\frac1p + \frac1q = 1$:
\begin{equation}
(\msr_p(X))^q = \max_{a,b\in\mathbb{R}}
\Bigl\{ b - \E\Bigl(\tfrac{q-1}{q^p} \bigl|(aX+b)_+-q\bigr|^p + (aX+b)_+
\Bigr)  \Bigr\}.
\label{msr repr}
\end{equation}

\noindent
2) For $p=1,2$:
\begin{equation}
\frac{1}{1+(\msrp(X,r))^p} = \min_{c\in\R} \E(1-cX)_+^p.
\label{msr12 repr}
\end{equation}
\end{theorem}

The main point about this theorem is that it
allows to reduce the problem of computing $\msrp$ as the supremum
over the set of random variables to the optimization  problem with one
or two real parameters and the  convex objective function. The
latter problem is much easier than the former one, since there exist
efficient algorithms of numerical convex optimization. This gives a
convenient way to compute $\msrp(X)$ (though only numerically, unlike the
standard Sharpe ratio). We'll also see that the representation is useful for
establishing some theoretical results about $\msrp$.

For the proof, we need the following auxiliary lemma.
\begin{lemma}
\label{sp repr}
Suppose $X\in\Lp$, $p\in[1,\infty)$, and $q$ is such that $\frac
1p + \frac 1q = 1$. Then
\[
\sigma_p(X) = \max\{\E (RX) \mid R\in \L^q,\; \E R=0, \|R\|_q \le 1 \}.
\]
\end{lemma}
\begin{proof}
Suppose $\sigma_p(X) = \|X-c^*\|_p$. By H\"older's inequality, for any
$R\in\L^q$ with $\E R = 0$ and $\|R\|_q \le 1$ we have
\[
\E (RX) = \E (R(X-c^*)) \le \|R\|_q \cdot \|X-c^*\|_p \le \|X-c^*\|_p.
\]
On the other hand, the two inequalities turn into equalities for
\[
R^* = \frac{ \sgn(X-c^*)\cdot|X-c^*|^{p-1}}{\|X-c^*\|_p^{p-1}}
\] 
and $R^*$ satisfies the above constraints.

\end{proof}

\begin{proof}[Proof of Theorem~\ref{th msr repr}]
Without loss of generality, assume $\E X = 1$.  First we're going to
show that $\msrp$ can be represented through the following optimization
problem:
\begin{equation}
\label{pr1 1}
\msrp(X) = \inf_{R\in \L^q} \{\|R\|_q \mid R\le1\ \as,\ \E R = 0,\ \E(RX)=1\}.
\end{equation}
In \eqref{sp def}, introduce the new  variables: $c=(\E Y)^{-1} \in \R$ and $Z
= cY \in \Lp$. Then
\[
\frac{1}{\msr_p(X)} = \inf_{\substack{Z\in \L^p\\ c\in\R}} \{\sigma_p(Z) \
\mid Z \le cX,\; \E Z =1\}.
\]
Consider the dual of the optimization problem in the RHS (see the Appendix for a brief overview of duality
methods in optimization). Define the dual objective function $g\colon
\L^q_+\times\R \to \R$ by
\[
g(u,v) = \inf_{\substack{Z\in\Lp\\c\in\R}} \{\sigma_p(Z) + \E(u(Z-cX)) - v(\E Z - 1) \}.
\]
The dual problem consists in maximizing $g(u,v)$ over all $u\in \L^q_+$,
$v\in\R$. We want to show that the strong duality takes place, i.e. that the
values of the primal and the dual problems are equal:
\[
\frac{1}{\msr_p(X)} = \sup_{\substack{u\in\L^q_+\\v\in \R}} g(u,v).
\]
To verify the sufficient condition for the strong duality from
Theorem~\ref{th duality}, introduce the optimal value function $\phi\colon \Lp\times\R\to[-\infty,\infty)$
\[
\phi(a,b) = \inf_{\substack{Z\in \L^p\\ c\in\R}} \{\sigma_p(Z) \
\mid Z - cX \le a,\; \E Z -1 = b\} 
\]
(obviously, $(\msrp(X))^{-1} = \phi(0,0)$). Observe that if a pair
$(Z_1,c_1)$ satisfies the constraints in $\phi(a_1,b_1)$ then the pair
$(Z_2,c_2)$ with
\[
c_2 = c_1 + b_2-b_1 + \E(a_1-a_2), \quad Z_2 = Z_1 +a_2- a_1 + (c_2-c_1)X, 
\]
satisfies the constraints in $\phi(a_2,b_2)$. Clearly, $\|Z_1-Z_2\|_p +
|c_1-c_2| = O(\|a_1-a_2\|_p + |b_1-b_2|)$, which implies that $\phi(a,b)$ is
continuous, so the strong duality holds.

Let us now transform the dual problem. It is obvious that if
$\E (uX) \neq 0$, then $g(u,v)=-\infty$ (minimize over $c$). For $u$ such
that $\E(uX)=0$, using the dual representation of $\sigma_p(X)$, we can
write
\[
g(u,v) = \inf_{Z\in\Lp} \sup_{R \in \mathcal{R}} \E(Z(R+u-v) + v)
\qquad\text{if }\E(uX) = 0,
\]
where $\mathcal{R} = \{R \in \L^q : \E R = 0,\; \|R\|_q \le 1\}$ is the dual
set for $\sigma_p$ from Lemma~\ref{sp repr}. Observe that the set
$\mathcal{R}$ is compact in the weak-$*$ topology by the Banach-Alaoglu
theorem. Consequently, by the minimax theorem (see Theorem~\ref{sion
theorem}), the supremum and infimum can be swapped. Then it is easy to see
that $g(u,v)>-\infty$ only if there exists $R\in\mathcal{R}$ such that
$R+u-v=0$ a.s., and in this case $g(u,v) = v$. Therefore, the dual problem
can be written as follows:
\[
\begin{split}
\frac{1}{\msrp(X)} &= \sup_{\substack{u\in \L^q\\v\in\R}} \{v \mid u\ge 0\
\as,\ \E(uX) = 0,\ v-u\in \mathcal{R}\} \\ &=
\sup_{R\in\mathcal{R}} \{\E (RX) \mid R \le \E(RX)\ \as\} \\ &=
\sup_{R\in\L^q} \{\E (RX) \mid R \le \E(RX)\ \as,\ \E R = 0,\ \|R\|_q \le 1\},
\end{split}
\]
where in the second equality we used that if $v-u=R\in\mathcal{R}$, then
the second constraint imply that $v = \E(RX)$ since it is assumed that $\E
X=1$. Now by changing the variable $R$ to $R/\E(RX)$ in the right-hand
side, we obtain representation \eqref{pr1 1}.

From \eqref{pr1 1}, it is obvious that for $p>1$
\begin{equation}
(\msrp(X))^q = \inf_{R\in \L^q} \{\E|R|^q \mid R\le1\ \as,\ \E R = 0,\
\E(RX)=1\}.
\label{pr1 2}
\end{equation}
We'll now consider the optimization problem dual to this one. Denote its
optimal value function by $\phi\colon \L^q\times\R \times \R \to \Re$. It
will be  more convenient to change the optimization variable $R$ here
by $1-R$ (which clearly doesn't change the value of $\phi$), so that
\[
\phi(a,b,c) = \inf_{R\in \L^q} \{\E|R-1|^q \mid R\ge a\ \as,\ \E R = 1+b,\ \E(RX)=c\}.
\]
 
Let us show that $\phi$ is continuous at zero. Denote by
$C(a,b,c)\subset \L^q$ the set of $R\in\L^q$ satisfying the constraints of the
problem. It will be enough to show that if $\|a\|_q,|b|,|c|$ are
sufficiently small then for any $R \in C(0,0,0)$ there exists
$\tilde R \in C(a,b,c)$ such that
$\|R-\tilde R\|_q \le (\|R\|_q+K) (\|a\|_q + |b| + |c|)$ and vice versa. Here $K$ is
some fixed constant.

Since $\P(X<0)\neq0$, there exists $\xi\in \L^\infty$ such that $\xi\ge 0$~a.s.
and $\E(\xi X) =-1$. If $R\in C(0,0,0)$, then one can take the required $\tilde R \in
C(a,b,c)$ in the form
\[
\tilde R =
\begin{cases}
a+ \lambda_1  R  + \lambda_2\xi, &\text{if }\E(aX)\ge 0,\\
a + \mu_1 R + \mu_2, &\text{if }\E(aX)<0,
\end{cases}
\]
where the non-negative constants $\lambda_1,\lambda_2,\mu_1,\mu_2$ can be easily
found from the constraint $\tilde R \in C(a,b,c)$, and it turns out that
$\lambda_1,\mu_1 = 1 + O(\|a\|_q+|b|+|c|)$ and $\lambda_2,\mu_2 = O(\|a\|_q+|b|+|c|)$. If $R\in C(a,b,c)$, then
take
\[
\tilde R =
\begin{cases}
\lambda_1(R-a+\lambda_2\xi), &\text{ if }c \ge \E(aX),\\
\mu_1(R-a+\mu_2), &\text{ if }c < \E(aX),
\end{cases}
\]
with $\lambda_i,\mu_i$ making $\tilde R \in C(0,0,0)$. 

Thus, the strong duality holds in \eqref{pr1 2} and we have
\begin{equation}
\msrp(X) = \sup_{\substack{u\in \L^q_+\\ v,w\in\R}} g(u,v,w)
\label{pr1 3}
\end{equation}
with the dual objective function $g\colon\L^q_+\times\R\times\R\to\Re$
\[
\begin{split}
g(u,v,w) &= \inf_{R\in \L^q} \E(|R|^q + R(u+v+wX) - u - w) \\ &=
-\E\Bigl(\tfrac{q-1}{q^p} |u + v + wX|^p + u+w\Bigr),
\end{split}
\]
where the second inequality is obtained by choosing $R$ which minimizes the
expression under the expectation for every random outcome.

Observe that for any fixed $v,w\in\R$ the optimal $u^*=u^*(v,w)$ in \eqref{pr1 3} can
be found explicitly: $u^* = (v + wX + q)_-$. Then by straightforward
algebraic transformation we obtain \eqref{msr repr}.

For $p=2$, from \eqref{msr repr} we get
\[
(\msr_2(X))^2 = \max_{a,b\in\R}\Bigl\{b - \frac14\E(aX+b)^2_+ - 1\Bigr\}
\]
It is easy to see that it is enough to maximize only over $b\ge 0$.
Maximizing over $b$ and introducing the variable $c=-\frac ab$, we obtain
representation \eqref{msr12 repr} for $p=2$.

To obtain representation \eqref{msr12 repr} for $p=1$, let's again consider
problem \eqref{pr1 1}. Similarly to \eqref{pr1 3} (the only change will be
to use $\|R\|_q$ instead of $\E|R|^q$), we can obtain that
\[
\msr_1(X) = \sup_{\substack{u\in \L^\infty_+\\v,w\in\R}} g(u,v,w), 
\]
where now we denote
\[
g(u,v,w) = \inf_{R\in\L^\infty} \{\|R\|_\infty + \E (R(u+v+wX) - u) - w \}.
\]
Observe that a necessary condition for $g(u,v,w)>-\infty$ is that
$\E|u+v+wX| \le 1$: otherwise take $\tilde R = c(\I(u+v+wX\le0) -
\I(u+v+wX> 0))$ and let $c\to\infty$. 
Under this condition we have $g(u,v,w) = -\E u - w$ since from H\"older's
inequality $|\E((\alpha +v +wX)R)| \le
\|R\|_\infty$ and therefore the infimum in $g$ is attained at $R=0$ a.s.
Consequently, the dual problem becomes
\begin{equation}
\msr_1(X) = -\inf_{\substack{u\in \L^\infty\\ v,w\in\R}}\{\E u + w \mid u\ge0\
\as,\ \E|u+v+wX| \le 1\}.
\label{pr1 4}
\end{equation}
Observe that the value of the infimum is non-positive, and so it is enough to
restrict the values of $w$ to $\R_-$ only. Let's fix  $v\in\R$, $w\in\R_-$ and find the optimal $u^* = u^*(v,w)$. Clearly, whenever
$v+wX(\omega) \ge 0$, it's optimal to take $u^*(\omega)=0$. Whenever
$v+wX(\omega) < 0$, we should have $u^*(\omega) \le |v+wX(\omega)|$, so
that $u(\omega) + v + wX(\omega) \le 0$ (otherwise, the choice  $u^*(\omega)
= |v+wX(\omega)|$ will be better).
Thus for the optimal $u^*$ 
\[
\E|u^* + v + wX| = \E|v+wX| -\E u^*.
\]
In particular, for the optimal $u^*$ the inequality in the second constraint in \eqref{pr1 4}
should be satisfied as the equality, since otherwise it would be possible to find
a smaller $u^*$. Observe that if $\E(v+wX)_+ > 1$, then no $u\in\L^\infty$ exists
which satisfies the constraint of the problem. On the other hand, if
$\E(v+wX)_+ \le 1$ then at least one such $u$ exists. Consequently, problem
\eqref{pr1 4} can be rewritten as follows:
\[
-\msr_1(X) = \inf_{v\in\R,w\in\R_-} \{\E|v+wX| + w - 1 \mid \E(v+wX)_+ \le 1\}.
\]
Clearly, $\E|v^*+w^*X|\le 0$ for the optimal pair $(v^*,w^*)$, so the
constraint should be satisfied as the equality (otherwise multiply both
$v,w$ by $1/\E|v+wX)_+$, which will decrease the value of the objective
function). By a straightforward transformation, we get
\[
1 + \msr_1(X) = \sup_{v\in\R,w\in\R_-} \{v \mid \E(v+wX)_+ = 1\}
\]
and introducing the new variable $c=w/v$, we obtain representation~\eqref{msr repr}.
\end{proof}

\subsection{Basic properties}
\begin{theorem}
\label{thm msr prop}
For any $p\in[1,\infty)$, the monotone Shape ratio in $\Lp$ satisfies the
following properties.
\begin{enumerate}[leftmargin=*,itemsep=0mm,topsep=0mm,parsep=0mm,label=(\alph*)]
\item (Quasi-concavity) For any $c\in \R$, the set $\{X \in \Lp: \msrp(X) \ge c\}$ is convex.
\item (Scaling invariance) $\msrp(\lambda X) = \msrp(X)$ for any real $\lambda>0$.
\item (Law invariance) If $X$ and $Y$ have the same distribution, then $\msrp(X) = \msrp(Y)$.
\item (2nd order monotonicity) If $X$ dominates $Y$ in the second stochastic
order, then $\msrp(X)\ge\msrp(Y)$.
\item (Continuity) $\msrp(X)$ is continuous with respect to $\L^p$-norm at any $X$ such
that  $\E X > 0$ and $\P(X<0)\neq 0$.
\end{enumerate}
\end{theorem}

Before proving this theorem, let us briefly discuss the properties in the context of
the portfolio selection problem.

The quasi-concavity implies that the monotone Sharpe ratio favors portfolio
diversification: if $\msrp(X)\ge c$ and $\msrp(Y)\ge c$, then $\msrp(\lambda
X + (1-\lambda)Y)\ge c$ for any $\lambda\in[0,1]$, where $\lambda X +
(1-\lambda)Y$ can be thought of as diversification between portfolios with
returns $X$ and $Y$. Note that the property of quasi-concavity is weaker
than concavity; it's not difficult to provide an example showing that
the monotone Sharpe ratio is not concave.

The scaling invariance can be interpreted as that the monotone Sharpe ratio
cannot be changed by leveraging a portfolio (in the same way as the standard
Sharpe ratio). Namely, suppose $X=R_x$, where $R_x = \sp xR$ is the return
of portfolio $x\in \R^{n+1}$ (as in Section~\ref{msr intro}), $\sum_i x_i =
1$. Consider a leveraged portfolio $\tilde x$ with $\tilde x_i = \lambda
x_i$, $i\ge 1$ and $\tilde x_0 = 1 - \sum \tilde x_i$, i.e. a portfolio
which is obtained from $x$ by proportionally scaling all the risky
positions. Then it's easy to see that $R_{\tilde x} = \lambda R_x$, and so
$\msrp(R_x) = \msrp(R_{\tilde x})$.

Law invariance, obviously, states that we are able to evaluate the
performance knowing only the distribution of the return. The interpretation
of the continuity property is also clear.

The 2nd order monotonicity means that $\msrp$ is consistent with preferences
of risk-averse investors. Recall that it is said that the distribution of a
random variable $X$ dominates the distribution of $Y$ in the 2nd stochastic
order, which we denote by $X\geqs Y$, if $\E U(X) \ge \E U(Y)$ for any
increasing concave function $U$ such that $\E U(X)$ and $\E U(Y)$ exist.
Such a function $U$ can be interpreted as a utility function, and then the
2nd order stochastic dominance means that $X$ is preferred to $Y$ by any
risk averse investor.

Regarding the properties from Theorem~\ref{thm msr prop}, let us also
mention the paper \cite{ChernyMadan08}, which studies performance measures
by an axiomatic approach in a fashion similar to the axiomatics of coherent
and convex risk measures. The authors define a performance measure (also
called an acceptability index) as a functional satisfying certain
properties, then investigate implications of those axioms, and show a deep
connection with coherent risk measures, as well as provide examples of
performance measures. The minimal set of four axioms a performance measure
should satisfy consists of the quasi-concavity, monotonicity, scaling
invariance and semicontinuity (in the form of the so-called Fatou property
in $\L^\infty$, as the paper \cite{ChernyMadan08} considers only functionals
on $\L^\infty$). In particular, the monotone Sharpe ratio satisfies those
axioms and thus provides a new example of a performance measure in the sense
of this system of axioms. It also satisfies all the additional natural
properties discussed in that paper: the law invariance, the arbitrage
consistency ($\msrp(X) = +\infty$ iff $X\ge 0$ a.s. and $\P(X>0)\neq 0$) and
the expectation consistency (if $\E X < 0$ then $\msrp(X) = 0$, and if $\E X
> 0$ then $\msrp(X) > 0$; this property is satisfied for $p>1$).

\begin{proof}[Proof of Theorem~\ref{thm msr prop}]
Quasi-concavity follows from that the $\Lp$--Sharpe ratio $S_p(X) = \frac{\E X}{\sigma_p(X)}$ is
quasi-concave. Indeed, if $S_p(X) \ge c$ and $S_p(Y) \ge c$, then
\[
S_p(\lambda X +(1-\lambda)Y) \ge \frac{\lambda \E X + (1-\lambda)\E Y}{\lambda
\sigma_p(X) + (1-\lambda) \sigma_p(Y)} \ge c
\]
for any $\lambda\in[0,1]$. 
Since $\msr_p$ is the maximum of $f_Z(X) = S_p(X-Z)$ over
$Z\in\L^p_+$, the quasi-concavity is preserved.

The scaling invariance is obvious. Since the expectation and the
$\Lp$-deviation are law invariant, in order to prove the law invariance of
$\msrp$, it is enough to show that the supremum in the definition of
$\msrp(X)$ can be taken over only $Y\le X$ which are measurable with respect
to the $\sigma$-algebra generated by $X$, or, in other words, $Y=f(X)$ for
some measurable function $f$ on $\R$. But this follows from the fact that if
for any $Y\le X$ one considers $\tilde Y = \E(Y\mid X)$, then $\tilde Y \le
X$, $\E(\tilde Y) = \E Y$ and $\sigma_p(Y) \le \sigma_p(Y)$, hence
$S_p(\tilde Y ) \ge S_p(Y)$.

To prove the 2nd order monotonicity, recall that another characterization of
the 2nd order stochastic dominance is as follows: $X_1\leqs X_2$ if and only
if there exist random variables $X_2'$ and $Z$ (which may be defined on a
another probability space) such that $X_2 \stackrel{d}{=} X_2'$, $X_1
\stackrel{d}{=} X_2' + Z$ and $\E(Z\mid X_2') \le 0$. Suppose $X_1 \leqs
X_2$. From the law invariance, without loss of generality, we may assume
that $X_1,X_2,Z$ are defined on the same probability space. Then for any
$Y_1\le X_1$ take $Y_2 = \E(Y_1 \mid X_2)$. Clearly, $Y_2 \le X_2$, $\E Y_2
= \E Y_1$ and $\sigma_p(Y_2) \le \sigma_p(Y_1)$. Hence $\sigma_p(X_1) \le
\sigma_p(X_2)$.

Finally, the continuity of $\msrp(X)$ follows from that the expectation and
the $\Lp$-deviation are uniformly continuous. 
\end{proof}

\section{Buffered probabilities}
In the paper \cite{RockafellarRoyset10} was introduced the so-called
buffered probability, which is defined as the inverse function of the
conditional value at risk (with respect to the risk level). The authors of
that and other papers (for example, \cite{DavisUryasev16}) argue that in
stochastic optimization problems related to minimization of probability of
adverse events, the buffered probability can serve as a better optimality
criterion compared to the usual probability.

In this section we show that the monotone Sharpe ratio is tightly related to
the buffered probability, especially in the cases $p=1,2$. In particular,
this will provide a connection of the monotone Share ratio with the
conditional value at risk. We begin with a review of the conditional value
at risk and its generalization to the spaces $\Lp$. Then we give a
definition of the buffered probability, which will generalize the one in
\cite{RockafellarRoyset10,MafusalovUryasev14} from $\L^1$ to arbitrary $\Lp$.

\subsection{A review of the conditional value at risk}
Let $X$ be a random variable, which describes loss. As opposed to the
previous section, now large values are bad, small values are good (negative
values are profits). For a moment, to avoid technical difficulties, assume
that $X$ has a continuous distribution.

Denote by $Q(X,\lambda)$ the $\lambda$-th quantile of the distribution of
$X$, $\lambda\in[0,1]$, i.e. $Q(X,\lambda)$ is a number $x\in\Re$, not
necessarily uniquely defined, such that $\P(X\le x) = \lambda$. The quantile
$Q(X,\lambda)$ is also called the value at risk\footnote{Some authors use
definitions of VAR and CVAR which are slightly different from the ones used
here: for example, take $(-X)$ instead of $X$, or $1-\lambda$ instead of
$\lambda$, etc.} (VAR) of $X$ at level $\lambda$, and it shows that in the
worst case of probability $1-\lambda$, the loss will be at least
$Q(X,\lambda)$. This interpretation makes VAR a sort of a measure of risk
(in a broad meaning of this term), and it is widely used by practitioners.

However, it is well-known that VAR lacks certain properties that one expects
from a measure of risk. One  of the most important drawbacks is that it
doesn't show what happens with probability less than $1-\lambda$. For
example, an investment strategy which loses \$1\,mln with probability 1\% and \$2\,mln with probability
0.5\% is quite different from a strategy which loses \$1\,mln and \$10\,mln with the same
probabilities, however they will have the same VAR at the 99\% level. Another
drawback of VAR is that it's not convex -- as a consequence, it may not
favor diversification of risk, which leads to concentration of risk (above
$1-\lambda$ level).

The conditional value at risk (CVAR; which is also called the average value
at risk, or the expected shortfall, or the superquantile) is considered as
an improvement of VAR. Recall that if $X\in\L^1$ and has a continuous
distribution, then CVAR of $X$ at risk level $\lambda\in[0,1]$ can be
defined as the conditional expectation in its right tail of probability
$1-\lambda$, i.e.
\begin{equation}
\cvar(X,\lambda) = \E(X \mid X> Q(X,\lambda))
\label{cvar def}
\end{equation}
We will also use the notation $\QQ(X,\lambda) = \cvar(X,\lambda)$ to
emphasize the
connection with quantiles.

CVAR provides a basic (and the most used) example of a coherent risk measure. The theory of risk
measures, originally introduced in the seminal paper \cite{ADEH99}, plays now a
prominent role in applications in finance. We are not going to discuss all
the benefits of using coherent (and convex) risk measures in optimization
problems;  a modern review of the main
results in this theory can be found, for example, in the monograph
\cite{FollmerSchied11}. 

Rockafellar and Uryasev \cite{RockafellarUryasev00} proved that CVAR admits the following
representation though the optimization problem
\begin{equation}
\QQ(X,\lambda) = \min_{c\in\R} \biggl(\frac{1}{1-\lambda} \E(X-c)_+ + c\biggr).
\label{cvar repr}
\end{equation}
Actually, this formula can be used as a general definition for CVAR, which
works in the case of any distribution of $X$, not necessarily continuous. The
importance of this representation is that it provides an efficient method to
compute CVAR, which in practical applications often becomes much faster
than e.g. using formula~\eqref{cvar def}. It also behaves ``nicely'' when CVAR
is used as a constraint or an optimality criterion in convex optimization
problems, for example portfolio selection. Details can be found in \cite{RockafellarUryasev00}.

Representation~\eqref{cvar repr} readily suggests how CVAR can be
generalized to ``put more weight'' on the right tail of the distribution of
$X$, which provides a coherent risk measure for the space $\Lp$.

\begin{definition*}
For $X\in\Lp$, define the $\Lp$-CVAR at level $\lambda\in[0,1)$ by
\[
\QQ_p(X,\lambda) = \min_{c\in\R} \biggl(\frac{1}{1-\lambda} \|(X-c)_+\|_p + c\biggr).
\]
\end{definition*}
The $\Lp$-CVAR was studied, for example, in the papers \cite{Krokhmal07,CheriditoLi09}. In particular,
in \cite{Krokhmal07}, it was argued that higher values of $p$ may provide better
results than the standard CVAR ($p=1$) in certain portfolio selection
problems. For us, the cases $p=1,2$ will be the most interesting due  the
direct connection with the monotone Sharpe ratio, as will be shown in the
next section.

It is known that the following dual representation holds for $\L^p$-CVAR,
which we will use below: for any $X\in \L^p$ and $\lambda\in[0,1)$
\begin{equation}
\QQ_p(X,\lambda) = \sup\bigl\{\E(RX) \mid R\in \L^q_+,\ \|R\|_q \le
(1-\lambda)^{-1},\ \E R =1\bigr\},
\label{cvar dual}
\end{equation}
where, as usual, $\frac1p+\frac1q =1$. This result is proved in \cite{CheriditoLi09}.

\subsection{The definition of buffered probability and its
representation}
Consider the function inverse to CVAR in $\lambda$, that is for a random
variable $X$ and $x\in\R$ define $\PP(X,x) = \lambda$ where $\lambda$ is
such that $\QQ(X,\lambda)=x$ (some care should be taken at points of
discontinuity, a formal definition is given below). In the papers
\cite{RockafellarRoyset10,MSU18,MafusalovUryasev14}, $\PP(X,x)$ was called
the ``buffered'' probability that $X>x$; we explain the rationale behind
this name below. At this moment, it may seem that from a purely
mathematical point of view such a simple operation as function inversion
probably shouldn't deserve much attention. But that's not the case if we
take applications into account. For this reason, before we give any
definitions, let us provide some argumentation why studying $\PP(X,x)$ may
be useful for applications.
  
In many practical optimization problems one may want to consider constraints
defined in terms of probabilities of adverse event, or to use those
probabilities as optimization criteria. For example, an investment fund manager may want to
maximize the expected return of her portfolio under the constraint that the
probability of a loss more than \$1\,mln should be less than 1\%; or an
engineer wants to minimize the construction cost of a structure
provided that the tension in its core part can exceed a critical
threshold with only a very small probability during its lifetime.

Unfortunately, the probability has all the same drawbacks as the value at
risk, which were mentioned above: it's not necessarily convex, continuous and doesn't provide
information about how wrong things can go if an adverse event
indeed happens. For those reasons, CVAR may be a better risk measure, which
allows to avoid some of the problems. For example, if using CVAR, the
above investor can reformulate her problem as maximization of the expected
return given that the average loss in the worst 1\% of cases doesn't exceed
\$1\,mln. However, such a setting of the problem  may be inconvenient, as CVAR ``speaks'' in terms of
quantiles, but one may need the answer in terms of probabilities. For
example, \$1\,mln may be value of liquid assets of the fund which can be
quickly and easily sold to cover a loss; so the manager must ensure that the loss
doesn't exceed this amount. But it is not clear how she can use the information about the
average loss which CVAR provides. A similar problem arises in the
example with an engineer.

In \cite{RockafellarRoyset10}, Rockafellar and Royset proposed the idea that the inverse of
CVAR may be appropriate for such cases: since quantiles and probabilities
are mutually inverse, and CVAR is a better alternative to quantiles, then
one can expect that the inverse of CVAR, the buffered probability, could be a better alternative to
probability. Here, we follow this idea.

Note that, in theory, it is possible to invert CVAR as a function in
$\lambda$, but, in practice, computational difficulty may be a serious
problem for doing that: it may take too much time to compute CVAR for a
complex system even for one fixed level of risk $\lambda$, so inversion,
which requires such a computation for several $\lambda$, may be not feasible
(and this is often the case in complex engineering or financial models).
Therefore, we would like to be able to work directly with buffered
probabilities, and have an efficient method to compute them. We'll see that
the representation given below turns out to give more than just an efficient
method of computation. In particular, in view of Section 2, it will show a
connection with the monotone Sharpe ratio, a result which is by no means
obvious.

The following simple lemma will be needed to show that it is possible to invert CVAR.

\begin{lemma}
\label{bpoe lemma}
For $X\in \L^p$, $p\in[1,\infty)$, the function $f(\lambda) = 
\QQ_p(X,\lambda)$ defined for $\lambda\in[0,1)$ has the following properties:
\begin{enumerate}[leftmargin=*,itemsep=0mm,topsep=1mm]
\item $f(0) = \E X$;
\item $f(\lambda)$ is continuous and non-decreasing;
\item $f(\lambda)$ is strictly increasing on the set $\{\lambda : f(\lambda)
< \esssup X\}$;
\item if $P:=\P(X=\esssup X) > 0$, then $f(\lambda) = \esssup X$ for $\lambda\in[1-P^{1/p},1)$.
\end{enumerate}
\end{lemma}
\begin{proof}
The first property obviously follows from the dual representation, and the
second one can be easily obtained from the definition. To prove the third
property, observe that if $\QQ_p(X,\lambda)<\esssup X$, then the minimum in
the definition is attained at some $c^*<\esssup X$. So, for any
$\lambda'<\lambda$ we have $\QQ_p(X,\lambda') \le
\frac{1}{1-\lambda'}\|(X-c^*)_+\|_p+c^* < \QQ_p(X,\lambda)$ using that
$\|(X-c^*)_+\|_p>0$.

Finally, the fourth property follows from that if $P> 0$, and, in
particular, $\esssup X < \infty$, then $\QQ_p(X,\lambda) \le \esssup X$ for
any $\lambda\in[0,1)$, as one can take $c=\esssup X$ in the definition. On
the other hand, for $\lambda_0 = 1-P^{\frac1p}$ we have that $R =
P^{-1}\I(X=\esssup X)$ satisfies the constraint in the dual representation
and $\E(RX) = \esssup X$. Hence $\QQ(X,\lambda_0) = \esssup X$, and then
$\QQ_p(X,\lambda) = \esssup X$ for any $\lambda\ge \lambda_0$ by the
monotonicity.
\end{proof}

\begin{definition*}
For $X\in\Lp$, $p\in[1,\infty)$, and $x\in\R$, set
\[
\PP_p(X,x) =
\begin{cases}
0, &\text{if }x> \esssup X,\\
(\P(X=\sup X))^{\frac1p}, &\text{if }x= \esssup X,\\
1-\mathbb Q^{-1}_p(X,x), &\text{if } \E X < x < \esssup X,\\
1, &\text{if }x \le \E X.
\end{cases}
\]
\end{definition*}
The ``main'' case in this definition is the third one. In particular, one
can see that for a random variable $X$ which has a distribution with a
support unbounded from above, the first and the second cases do not realize.
Figure~\ref{fig bpoe} schematically shows the relation between the quantile
function, the CVAR, the probability distribution function, and the buffered
probability. In particular, it is easy to see that always $\PP_p(X,x) \ge
\P(X>x)$. According to the terminology of \cite{RockafellarRoyset10}, the
difference between these two quantities is a ``safety buffer'', hence the
name buffered probability.

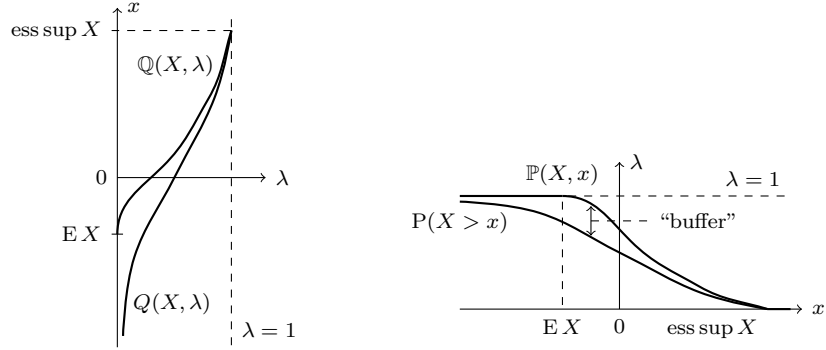
\begin{figure}[h]
\begin{tikzpicture}[scale=1.5]
\draw[thick] (0,-0.5) to[out=90, in=225] (0.4,0.1) to[out=45,in=240]
(0.8,0.7) to[out=60,in=253] node[left,xshift=0.1em] {\scriptsize
$\QQ(X,\lambda)$} (1,1.3);
\draw[thick] plot [smooth,tension=0.7] coordinates { (0.05,-1.4) (0.15,-0.7)
(0.5,0) (0.85,0.7) (1,1.3) };
\draw[->] (0,0) node[left] {\scriptsize $0$} -- (1.3,0) node[right] {\scriptsize $\lambda$};
\draw[->] (0,-1.5) -- (0,1.5) node[right](x) {\scriptsize $x$};
\draw[dashed] (1,1.4) -- (1,-1.5) node[above right] {\scriptsize $\lambda=1$};;
\draw (0.05,-1.3) node[above right] {\scriptsize $Q(X,\lambda)$};
\draw (0.05,-0.5)--(-0.05,-0.5) node[left] {\scriptsize $\E X$};
\draw [dashed] (1,1.3)--(-0.05,1.3) node[left] {\scriptsize $\esssup X$};
\end{tikzpicture}
\hspace{10mm}
\begin{tikzpicture}[scale=1.5]
\draw[thick] plot [smooth,tension=0.7] coordinates { (-1.4,0.95) (-0.7,0.85)
(0,0.5) (0.7,0.15) (1.3,0.0) };
\draw[thick] (-0.5,1) to[out=0, in=135] (0.1,0.6) to[out=-45,in=150] (0.7,0.2) to[out=-30,in=163]
(1.3,0);
\draw[thick] (-1.4,1)--(-0.5,1);

\draw[->] (0,0) node[below] {\scriptsize $0$} -- (0,1.3) node[right] {\scriptsize $\lambda$};
\draw[->] (-1.4,0) -- (1.6,0) node[right] {\scriptsize $x$};
\draw[dashed] (-0.45,1) -- (1.5,1) node[above left] {\scriptsize $\lambda=1$};

\draw (-0.5,1) node[above] {\scriptsize $\mathbb P(X,x)$};
\draw (-1.4,0.95) node[below] {\scriptsize $\P(X>x)$};
\draw[dashed] (-0.5,0.95)--(-0.5,0) node[below] {\scriptsize $\E X$};
\draw (1.3,0) node[below left] {\scriptsize $\esssup X$};
\draw[thick] (1.3,0.0)--(1.5,0.0);
\draw[<->] (-0.25,0.65)--(-0.25,0.91);
\draw[dashed] (-0.25,0.78) -- (0.25,0.78) node[right] {\scriptsize ``buffer''};
\end{tikzpicture}
\centering
\caption{Left: quantile and distribution functions. 
Right: complementary probability distribution function and buffered
probability $\PP(X,x)$. In this example, $\esssup X < \infty$, but
$\P(X=\esssup X) =0$, so $\PP(X,x)$ is continuous everywhere.}
\label{fig bpoe}
\end{figure}

\begin{theorem}
For any $X\in \L^p$
\begin{equation}
\PP_p(X,x) = \min_{c\ge 0} \|(c(X-x)+1)_+\|_p.
\label{bpoe formula}
\end{equation}
\end{theorem}
\begin{proof}
For the case $p=1$ this result was proved in \cite{MafusalovUryasev14}. Here, we follow the
same idea, but for general $p\in[1,\infty)$. Without loss of generality, we can assume
$x=0$, otherwise consider $X-x$ instead of $X$.

\textit{Case 1:} $\E X < 0$, $\esssup X > 0$. By Lemma~\ref{bpoe lemma} and the
definition of $\QQ_p$ we have
\[
\begin{split}
\PP_p(X,0) &= \min\{\lambda \in(0,1) \mid \QQ_p(X,1-\lambda) = 0 \} \\
&= \min_{\lambda \in(0,1)}\{\lambda \mid \min_{c\in\R}(\tfrac1\lambda \|((X+c)_+\|_p - c)
= 0\} \\
&=\min_{\substack{\lambda \in(0,1)\\c\in\R}}\{\lambda \mid 
\|(X+c)_+\|_p  =\lambda c\}.
\end{split}
\]
Observe that the minimum here can be computed only over $c> 0$ (since for
$c\le 0$ the constraint is obviously not satisfied). Then dividing the both
parts of the equality in the constraint by $c$ we get
\[
\PP_p(X,0) = \min_{c> 0} \|(c^{-1} X+1)_+\|_p,
\]
which is obviously equivalent to~\eqref{bpoe formula}.

\textit{Case 2:} $\E X \ge 0$. We need to show that $\min\limits_{c\ge 0}
\|(cX+1)_+\|_p = 1$. This clearly follows from that for any $c\ge 0$ we have $
\min\limits_{c\ge 0} \|(cX+1)_+\|_p \ge \min\limits_{c\ge 0}\E(cX+1) = 1$.

\textit{Case 3:} $\esssup X = 0$. Now $\|(cX+1)_+\|_p \ge \P(X=0)^{1/p}$ for
any $c\ge 0$, while $\|(c X + 1)_+\|_p \to \P(X=0)^{1/p}$ as $c\to +\infty$.
Hence $\min\limits_{c\ge 0} \|(cX+1)_+\|_p = \P(X=0)^{1/p}$ as claimed.

\textit{Case 4:} $\esssup X < 0$. Similarly, $\|(c X + 1)_+\|_p \to 0$ as $c\to +\infty$.
\end{proof}

From formula \eqref{bpoe formula}, one can easily see the connection between
the monotone Sharpe ratio and the buffered probability for $p=1,2$: for any
$X\in \L^p$
\[
\frac{1}{1+(\msrp(X))^p} = (\PP_p(-X,0))^p.
\]
In particular, if  $X$ is as the return of a portfolio, then a portfolio selection problem where one wants to maximize the
monotone Sharpe ratio of the portfolio return becomes equivalent to the minimization of the
buffered probability that $(-X)$ exceeds $0$, i.e. the buffered probability
of loss. This is a nice (and somewhat unexpected) connection between the
classical portfolio theory and modern developments in risk evaluation!

One can ask a question whether a similar relation between $\PP_p$ and $\msrp$ holds for other values of $p$.
Unfortunately, in general, there seems to be no simple formula connecting
them. It can be shown that they can be represented as the following
optimization problems:
\[
\begin{aligned}
&\msrp(X) = \min_{R\in \L^q_+}\{ \|R-1\|_q \mid \E R = 1,\ \E(RX) = 1\},\\
&\PP_p(X,0) = \min_{R\in \L^q_+}\{ \|R\|_q \mid \E R = 1,\ \E(RX) = 1\},
\end{aligned}
\]
which have the same constraint sets but different objective functions.
The first formula here easily follows from~\eqref{pr1 1}, the second one can
be obtained using the dual representation of CVAR \eqref{cvar dual}.

\subsection{Properties}
In this section we investigate some basic properties of $\PP_p(X,x)$ both in
$X$ and $x$, and discuss its usage in portfolio selection problem. One of
the main points of this section is that buffered probabilities (of loss) can
be used as optimality criteria, similarly to monotone Sharpe ratios (and in
the cases $p\neq1,2$ they are more convenient due to s simpler
representation).

\begin{theorem} Suppose $X\in\Lp$, $x\in\R$ and $p\in[1,\infty)$. Then
$\PP_p(X,x)$ has the following properties. 
\begin{enumerate}
\item The function $x\mapsto \PP_p(X,x)$ is continuous and strictly decreasing on $[\E X,
\esssup(X))$, and non-increasing on the whole $\R$.
\item The function  $X \mapsto \PP_p(X,x)$ is quasi-convex, law invariant, 2nd order monotone, 
continuous with respect to the $\Lp$-norm, and concave with respect to mixtures of distributions.
\item The function $p\mapsto \PP_p(X,x)$ is non-decreasing in $p$.
\end{enumerate}
\end{theorem}

For $p=1$, similar results can be found in \cite{MafusalovUryasev14}; the proofs are
similar as well (except property 3, but it obviously
follows from the Lyapunov inequality), so we do not provide them here.

Regarding the second property note that despite $\PP_p(X,x)$ is quasi-convex in $X$, it's not convex
in $X$ as the following simple example shows: consider $X\equiv 2$ and $Y\equiv -1$; then
$\PP((X+Y)/2, \;0) = 1 \not\le \tfrac12 = \tfrac 12 \PP(X,0) +
\tfrac12\PP(Y,0)$.

Also recall that the mixture of two distributions
on $\R$ specified by their distribution functions $F_1(x)$ and $F_2(x)$ is
defined as the distribution $F(x) = \lambda F_1(x) + (1-\lambda)F_2(x)$ for
any fixed $\lambda\in[0,1]$. We write
$X\stackrel{d}{=} \lambda X_1 \oplus (1-\lambda)X_2$ if the distribution of
a random variable $X$ is the mixture of the distributions of $X_1$ and
$X_2$. If $\xi$ is a random variable taking values $1,2$ with probabilities
$\lambda,1-\lambda$ and independent of $X_1,X_2$, then clearly
$X\stackrel{d}{=} X_{\xi}$. Concavity of $\P_p(x,x)$ with respect to
mixtures of distributions means that
$
\PP_p(X, x) \ge \lambda\PP_p(X_1,x) + (1-\lambda)\PP_p(X_2,x).
$

Now let's look in more details on how a simple portfolio selection problem can be
formulated with $\PP_p$. Assume the same setting as in Section~\ref{msr intro}: $R$
is a $(n+1)$-dimensional vector of asset returns, the first asset is
riskless with the rate of return $r$, and the other $n$ assets are risky
with random return in $\Lp$. Let $R_x = \sp xR$ denote the return of a
portfolio $x\in\R^{n+1}$, and $\delta>0$ be a fixed number, a required
expected return premium. Consider the following optimization problem:
\begin{equation}
\begin{aligned}
\text{minimize}\quad& \PP_p(r-R_x, \;0)\text{ over }x\in \R^{n+1}\\
\text{subject to}\quad&
\E(R_x-r) = \delta,\\
& {\textstyle \sum_i} x_i=1.
\end{aligned}
\label{p opt}
\end{equation}
In other words, an investor wants to minimize the buffered probability that
the return of her portfolio will be less than the riskless return subject to
the constraint on the expected return. Denote the vector of adjusted risky
returns $\overline R = (R_1-r,\ldots,R_n-r)$, and the risky part of the
portfolio $\overline x = (x_1,\ldots,x_n)$. Using the representation of
$\PP_p$, the problem becomes
\begin{equation}
\begin{aligned}[c]
\text{minimize}\quad& \E(1-\sp{\overline x}{\overline R})_+^p\text{ over
}\overline x\in\R^n\\
\text{subject to}\quad&
\E \sp{\overline x}{\overline R} \ge 0.
\end{aligned}
\label{p opt 2}
\end{equation}
If we find a solution $\overline x^*$ of this problem, then the optimal portfolio in
problem~\eqref{p opt} can be found as follows:
\[
x^*_i = \frac{\delta \overline x^*_i}{\E \sp{\overline x^*}{\overline R}},
\; i=1,\ldots,n, \qquad x^*_0 = 1 - {\textstyle \sum\limits_{i=1}^n x^*_i}.
\] 
Moreover, observe that the constraint $\E \sp{\overline x}{\overline R} \ge
0$ can be removed in \eqref{p opt 2} since the value of the objective
function is not less than 1 in the case if $\E \sp{\overline x}{\overline R} < 0$, which is not
optimal. Thus, \eqref{p opt 2} becomes an unconstrained problem.

\section{Dynamic problems}

This section illustrates how the developed theory can be used
to give new elegant solutions of dynamic portfolio selection problems
when an investor can continuously trade in the market. The results we
obtain are not entirely new, but their proofs are considerably shorter and simpler
than in the literature.

\subsection{A continuous-time market model and two investment problems}
Suppose there are two assets traded in the market: a riskless asset with
price $B_t$ and a risky asset with price $S_t$ at time $t\in[0,\infty)$. The
time runs continuously. Without loss of generality, we assume
$B_t\equiv 1$. The price of the risky asset is modeled by a geometric
Brownian motion with constant drift $\mu$ and volatility $\sigma$, i.e.
\[
S_t = S_0 \exp\biggl(\sigma W_t + \Bigl(\mu-\frac{\sigma^2}{2} t \Bigr)\biggr), \qquad t\ge 0,
\]
where $W_t$ is a Brownian motion (Wiener process). Without loss of
generality, $S_0=1$. It is well-known that the process $S_t$ is
the unique strong solution of the stochastic differential equation (SDE)
\[
d S_t = S_t(\mu dt + \sigma dW_t).
\]

We consider the following two problems of choosing an optimal investment
strategy in this market model.

\medskip \textit{Problem 1.} Suppose a trader can manage her portfolio
dynamically on a time horizon $[0,T]$. A trading strategy is identified with a
scalar control process $u_t$, which is equal to the amount of money invested
in the risky asset at time $t$. The amount of money $v_t$ is invested in the
riskless asset. The value $X_t^u=u_t+v_t$ of the  portfolio with
the starting value $X_0^u=x_0$ satisfies the controlled SDE
\begin{equation}
dX_t^u = u_t(\mu dt + \sigma dW_t),\qquad X_0^u = x_0.
\label{control SDE}
\end{equation}
This equation is well-known and it expresses the assumption that the trading strategy is self-financing,
i.e. it has no external inflows or outflows of capital. Note that $v_t$
doesn't appear in the equation since it can be uniquely recovered as $v_t =
X_t^u - u_t$. 

To have $X^u$ correctly defined, we'll assume that $u_t$ is predictable with
respect to the filtration generated by $W_t$ and
$\E \int_0^T u_t^2 dt < \infty$. We'll also need to impose the following
mild technical assumption:
\begin{equation}
\E \exp\biggl( \frac{\sigma^2p^2}2 \int_0^T \frac{u_t^2}{(1-X_t^u)^2} dt\biggr) <
\infty.
\label{4.1.2}
\end{equation}
The class of all the processes $u_t$ satisfying these assumptions will be denoted by
$\mathcal{U}$. Actually, it can be shown that \eqref{4.1.2}
can be removed without changing the class of optimal strategies in the
problem formulated below, but to keep
the presentation simple, we will require it to hold.

The problem consists in minimizing the buffered probability of loss by time $T$.
So the goal of the trader is to  solve the following control problem with some fixed $p\in(1,\infty)$:
\begin{equation}
V_1 = \inf_{u\in\mathcal{U}} \PP_p(x_0-X_T^u,0).
\label{v1}
\end{equation}
For $p=2$ this problem is equivalent to the problem of maximization of
the monotone Sharpe ratio $\msrp(X_T^u - x_0)$. Moreover, we'll also show
that the same solution is obtained in the problem of maximization of the
standard Sharpe ratio, $S(X_T^u-x_0)$. Note that we
don't consider the case $p=1$.

From \eqref{control SDE} and \eqref{v1}, it is clear that without loss of generality we can
(and will) assume $x_0=0$. It is also clear that there is no unique solution
of problem \eqref{v1}: if some $u^*$ minimizes $\PP_p(x_0-X_T^u,0)$ then so
does any $u_t = c u_t^*$ with a constant $c>0$. Hence, additional constraints
have to be imposed if one wants to have a unique solution, for example a
constraint on the expected return like $\E X_T^u = x_0 + \delta$. This is similar
to  the standard Markowitz portfolio selection problem, as discussed in
Section~\ref{msr intro}.

\medskip
\textit{Problem 2.} Suppose at time $t=0$ a trader
holds one unit of the risky asset with the starting price $S_0=1$ and wants to sell it better than some goal
price $x\ge1$. The asset is indivisible (e.g. a house) and can be sold only at
once.

A selling strategy is identified with a Markov time of the process $S_t$.
Recall that a random variable $\tau$ with values in $[0,\infty]$ is called a
Markov time if the random event $\{\tau\le t\}$ is in the $\sigma$-algebra
$\sigma(S_r;\ r\le t)$ for any $t\ge 0$. The notion of a stopping time
reflects the idea that no information about the prices in the future can be
used at the moment when the trader decides to sell the asset. The random
event $\{\tau=\infty\}$ is interpreted as the situation when the asset is
never sold, and we set $S_\infty := 0$. We'll see that the optimal strategy
in the problem we formulate below will not sell the asset with positive
probability.

Let $\mathcal{M}$ we denote the class of all Markov times of the process
$S_t$. We consider the following optimal stopping problem for
$p\in(1,\infty)$:
\begin{equation}
V_2 = \inf_{\tau \in\mathcal{M}} \PP_p(x-S_\tau,0),
\label{v2}
\end{equation}
i.e. minimization of the buffered probability to sell worse then for the
goal price $x$. Similarly to Problem 1, in the case $p=2$, it'll be shown
that this problem is equivalent to maximization of the monotone Sharpe ratio
$\msr_2(S_\tau-x)$, and the optimal strategy also maximizes the standard
Sharpe ratio.

\subsection{A brief literature review}
Perhaps, the most interesting thing to notice about the two problems is that they
are ``not standard'' from the point of view of the stochastic
control theory for diffusion processes and Brownian motion. Namely, they
don't directly reduce to solutions of some PDEs, which for
``standard'' problems is possible via the Hamilton--Jacobi--Bellman equation. In
Problems 1 and 2 (and also in the related problems of maximization of the
standard Sharpe ratio), the HJB equation cannot be written because the
objective function is not in the form of the expectation of some functional
of the controlled process, i.e. not $\E F(X_r^u;\ r\le T)$ or $\E F(S_r;\
r\le \tau)$. Hence, another approach is needed to solve them.

Dynamic problems of maximization of the standard Sharpe ratio and related
problems with mean-variance optimality criteria have been studied in the
literature by several authors. We just briefly mention several of them.

Richardson \cite{Richardson89} was, probably, the first who solved a
dynamic portfolio selection problem under a mean-variance criterion (the earlier
paper of White \cite{White74} can also be mentioned); he used ``martingale''
methods. Li and Ng \cite{LiNg00} studied a multi-asset mean-variance selection
problem, which they solved through auxiliary control problems in the
standard setting. A similar approach was also used in the recent papers by
Pedersen and Peskir \cite{PedersenPeskir16,PedersenPeskir17}. The first of
them provides a solution for the optimal selling problem (an analogue of our
Problem 2), the second paper solves the portfolio selection problem (our
Problem 1). There are other results in the literature, a comprehensive
overview can be found in the above-mentioned papers by Pedersen and Peskir,
and also in the paper \cite{CLWZ12}.

It should be also mentioned that a large number of papers study the
so-called problem of time inconsistency of the mean-variance and similar
optimality criteria, which roughly means that at a time $t>t_0$ it turns out
to be not optimal to follow the strategy, which was optimal at time $t_0$.
Such a contradiction doesn't happen in standard control problems for Markov
processes, where the Bellman principle can be applied, but it is quite
typical for non-standard problems. Several approaches to redefine the notion
of an optimal strategy that would take into account time inconsistency are
known: see, for example, the already mentioned papers
\cite{PedersenPeskir16,PedersenPeskir17,CLWZ12} and the references therein.
We will not deal with the issue of time inconsistency (our solutions are
time inconsistent).

Compared to the results in the literature, the solutions of Problems 1 and 2
in the case $p=2$ readily follows from earlier results (e.g. from
\cite{Richardson89,PedersenPeskir16,PedersenPeskir17}); the other cases can
be also studied by previously known methods. Nevertheless, the value of this
paper is in the new approach to solve them through the monotone Sharpe ratio
and buffered probabilities. This approach seems to be simpler than previous
ones (the reader can observe how short the solutions presented below
compared to \cite{PedersenPeskir16,PedersenPeskir17}) and promising
for more general settings.

\subsection{Solution of Problem 1}
\begin{theorem}
The class of optimal control strategies in Problem 1 is given by
\[
u_t^{c} = \frac{\mu}{\sigma^2(p-1)}(c - X_t^{u^c}),
\]
where $c>0$ can an arbitrary constant. The process $Y_t^{u^c} = c -
X_t^{u^c}$ is a geometric Brownian motion satisfying the SDE
\[
\frac{d Y^{u^c}_t}{Y^{u^c}_t} =- \frac{\mu^2}{\sigma^2 (p-1)}  dt -
\frac{\mu}{\sigma(p-1)} dW_t , \qquad Y_0^{u^c} = c.
\]
\end{theorem}
\begin{proof}
Assuming $x_0=0$, from the representation of $\PP_p(X,x)$ we have
\begin{equation}
V_1 = \min_{c\ge 0} \min_{u\in\mathcal{U}} \| (1- cX_T^u)_+\|_p =
\min_{u\in\mathcal{U}} \| (1- X_T^u)_+\|_p,
\label{4.1.1}
\end{equation}
where in the second equality we used that the constant $c$ can be included
in the control, since $cX^u = X^{cu}$. Denote $\tilde X_t^u = 1 - X_t^u$, so
that the controlled  process $\tilde X^u$ satisfies the equation
\[
d \tilde X_t^u = -\mu u_t dt - \sigma u_t d W_t, \qquad \tilde X_0^u = 1.
\]
Then 
\begin{equation}
V_1^p = \min_{u\in\mathcal{U}} \E| \tilde X_T^{u}|^p,
\label{4.1.1a}
\end{equation}
where  $(\cdot)_+$ from \eqref{4.1.1} was removed  since it is obvious that as
soon as $\tilde X_t^u$ reaches zero, it is optimal to choose $u\equiv 0$
afterwards, so the process stays at zero.

Let $v_t = v_t(u) = -u_t/\tilde X_t^u$. Then for any $u\in \mathcal{U}$ we
have
\[
\E |\tilde X_T^u|^p = \E \Bigl\{Z_T \exp\Bigl({\textstyle \int_0^T} \bigl(\mu
p v_s + \tfrac12\sigma^2(p^2-p) v_s^2 \bigr)ds\Bigr)\Bigr\},
\]
where $Z$ is the stochastic exponent process $Z= \mathcal{E}(\sigma p v)$. From
Novikov's condition, which holds due to \eqref{4.1.2}, $Z_t$ is a martingale
and  $E Z_T = 1$. By
introducing the new measure $Q$ on the $\sigma$-algebra $\mathcal{F}_T = \sigma(W_t,\
t\le T)$ with the density $dQ = Z_T dP$ we obtain
\[
\E |\tilde X_T^u|^p = \E^Q\Bigl\{ \exp\Bigl({\textstyle \int_0^T} \bigl(\mu
p v_s + \tfrac12\sigma^2(p^2-p) v_s^2 \bigr)ds\Bigr)\Bigr\}.
\]
Clearly, this expression can be minimized by minimizing the integrand for
each $t$, i.e. by
\[
v^*_t = -\frac{\mu}{\sigma^2(p-1)} \text{ for all }t\in[0,T].
\]
Obviously, it satisfies condition \eqref{4.1.2}, so
the corresponding control process 
\[
u_t^* = \frac{\mu}{\sigma^2(p-1)} \tilde X_t^u = \frac{\mu}{\sigma^2(p-1)} (1-X_t^u)
\]
is optimal in problem \eqref{4.1.1a}. Consequently, any control process
$u^c_t = c u_t^*$, $c>0$, will be optimal in \eqref{v1}. Since $X_t^{u^c} = c
X_t^{u^*}$, we obtain the first claim of the theorem. The representation for
$Y_t^{u^c}$ follows from straightforward computations.
\end{proof}

\begin{corollary}
Let $u^* = \frac{\mu}{\sigma^2}(c- X_t^u)$, with some $c>0$, be an optimal
control strategy in problem \eqref{v1} for $p=2$. Then the standard Sharpe
ratio of $X_T^{u^*}$ is equal to its monotone Sharpe ratio,
$S(X_T^{u^*}) = \msr_2(S_T^{u^*})$.

In particular, $u^*$ also maximizes
the standard Sharpe ratio of the return $X_T^u$, i.e. $S(X_T^u) \le
S(X_T^{u^*})$ for any
$u\in\mathcal{U}$.
\end{corollary}
\begin{proof}
Suppose there is $Y\in \L^2$ such that $Y\le X_T^{u^*}$ and
$S(Y)> S(X_T^{u^*})$. It is well-known that the market model we consider is
no-arbitrage and complete. This implies that there exists $y_0<0$ and a
control $u_t$ such that $X_0^{u} =y_0$ and $X_T^u = Y$. The initial capital
$y_0$ is negative, because otherwise an arbitrage opportunity can be
constructed. But then the capital process $\tilde X_t = X_t^{u}-y_0$ would have a
higher Sharpe ratio than $Y$ and hence a higher monotone Sharpe ratio than
$X^{u*}_T$. A contradiction. This proves the first claim of the corollary,
and the second one obviously follows from it.
\end{proof}

\subsection{Solution of Problem 2}
We'll assume that $x\ge 1$, $\mu\in \R$, $\sigma>0$, $p>1$ are fixed throughout and use
the following auxiliary notation:
\[
\gamma = \frac{2\mu}{\sigma^2}, \qquad C(b) =
\biggl(\frac{b}{1+\frac{x}{b-x}(1-b^{1-\gamma})^{\frac1{p-1}}}-x\biggr)^{-1}\text{
for } b\in[x,\infty).
\]
\begin{theorem}
\label{v2 thm}
The optimal selling time $\tau^*$ in problem \eqref{v2} is as follows.
\begin{enumerate}[leftmargin=*,itemsep=0mm,topsep=1mm]
\item If $\mu\le 0$, then $\tau^*$ can be any Markov time:
$\PP_p(x-S_\tau,0)=1$ for any $\tau\in\mathcal{M}$.

\item If $\mu\ge \frac{\sigma^2}{2}$, then $S_t$ reaches any level $x'> x$ with
probability $1$ and any stopping time of the form
$\tau^* = \inf\{t\ge 0 : S_t = x'\}$ is optimal.

\item If $0<\mu<\frac{\sigma^2}{2}$, then the optimal stopping time is
\[
\tau^* = \inf\{t\ge 0 : S_t \ge b^*\},
\]
where $b^*\in[x,\infty)$ is the point of minimum of the  function
\[
f(b) = ((1+C(b)(x-b))^p b^{\gamma-1} + (1+C(b)x)^p (1-b^{\gamma-1}), \qquad
b\in[x,\infty),
\]
and we set $\tau^* = +\infty$ on the random event $\{S_t < b\text{ for all }
t\ge 0\}$.
\end{enumerate}
\end{theorem}

Observe that if $0<\mu<\frac{\sigma^2}{2}$, i.e. $\gamma\in(0,1)$, then the
function $f(b)$ attains its minimum on $[x,\infty)$, since it is continuous
with the limit values $f(x)=f(\infty)=1$.

\begin{proof}
From the representation for $\PP_p$ we have
\[
V_2^p = \inf_{c\ge 0} \inf_{\tau\in\mathcal{M}}  \E|(1  +c(x-S_\tau))_+|^p.
\]
Let $Y_t^c = 1  +c(x-S_\tau)$. Observe that if $\mu\le 0$, then $Y_t^c$ is a
submartingale for any $c\ge 0$, and so by Jensen's inequality $|(Y_t)_+)|^p$ is a
submartingale as well. Hence for any $\tau\in\mathcal{M}$ we have
$\E|(1+c(x-S_\tau))_+|^p\ge 1$, and then $V_2=1$.

If $\mu \ge \frac{\sigma^2}{2}$, then from the explicit
representation $S_t = \exp(\sigma W_t + (\mu - \frac{\sigma^2}{2})t)$ one
can see that $S_t$ reaches any level
$x'\ge 1$ with probability 1 (as the Brownian motion $W_t$ with non-negative drift
does so). Then for any $x'>x$ we have $\PP_p(x-S_{\tau_{x'}},0)=0$, where
$\tau_{x'}$ is the first moment of reaching $x'$.

In the case $\mu\in (0,\frac{\sigma^2}2)$, for any $c\ge 0$, consider the
optimal stopping problem
\[
V_{2,c} = \inf_{\tau\in\mathcal{M}}  \E|(1  +c(x-S_\tau))_+|^p.
\]
This is an optimal stopping problem for a Markov process $S_t$. From the
general theory (see e.g. \cite{PeskirShiryaev06}) it is well known that the optimal stopping time here is of the
threshold type:
\[
\tau^*_c = \inf\{t\ge 0 : S_t \ge b_c\},
\]
where $b_c\in [x, x+\frac1c]$ is some optimal level, which has to be found. Then the distribution
of $S_{\tau_c^*}$ is binomial: it assumes only two values $b_c$ and $0$ with
probabilities $p_c$ and $1-p_c$, where $p_c=b_c^{\gamma-1}$ as can be easily
found from the general formulas for  boundary crossing probabilities for a
Brownian motion with drift. Consequently,
\[
V_2^p = \inf_{b\ge x} \inf_{c\le \frac1{(b-x)}} \Bigl((1+c(x-b))^p b^{\gamma-1} + (1+cx)^p (1-b^{\gamma-1})\Bigr).
\]
It is straightforward to find that for any $b\ge x$ the optimal $c^*(b)$ is given
by $c^*(b) =C(b)$, which proves the claim of the theorem.
\end{proof}

\begin{corollary}
Assume $\mu\in(0,\frac{\sigma^2}{2})$ and $p=2$. Let $\tau^*$ denote the
optimal stopping time from Theorem~\ref{v2 thm}. 
Then the standard Sharpe ratio of $S_{\tau^*}-x$ is equal to its monotone
Sharpe ratio, $S(S_{\tau^*}-x) = \msr_2(S_{\tau^*}-x)$. In particular,
$\tau^*$ also maximizes the standard Sharpe ratio of $S_{\tau}-x$, i.e.
$S(S_\tau-x) \le S(S_{\tau^*}-x)$ for any $\tau\in\mathcal{M}$.

Moreover, in this case the optimal
threshold $b^*$ can be found as the point of maximum
of the function
\[
g(b) = \frac{b^{\gamma}-x}{b^{\frac{\gamma+1}{2}}(1-b^{\gamma-1})^{\frac12}}.
\]
\end{corollary}
\begin{proof}
Suppose $Y\le S_{\tau^*} -x$. As shown above, it is enough to consider only $Y$
which are measurable with respect to the $\sigma$-algebra generated by
the random variable $S_{\tau^*}$. Since $S_{\tau^*}$ has a binomial
distribution, then $Y$ should also have a binomial distribution, assuming
values $y_1 \le b^*-x$ and $y_2\le-x$ with the same probabilities $(b^*)^{\gamma-1}$ and
$1-(b^*)^{\gamma-1}$ as $S_{\tau^*}$ assumes the values $b^*$ and 0. Using
this, it is now not difficult to see that $S(Y) \le S(S_{\tau^*}-x)$, which
proves the first claim.

The second claim follows from that for any stopping time of the form $\tau_b
= \{t\ge 0: S_t = b\}$, $b\in[x,\infty)$ we have $S(S_{\tau_b} - x) = g(b)$.
\end{proof}

\appendix
\section{Appendix}
This appendix just reminds some facts from convex optimization and related results which were used
in the paper.

\subsection{Duality in optimization}

Let $\Z$ be a topological vector space and $f(z)$ a real-valued function on $\Z$. 
Consider the optimization problem
\begin{equation}
\text{minimize }f(z) \text{ over }z\in \mathcal{Z}.
\label{opt problem}
\end{equation}
A powerful method to analyze such an optimization problem consists in
considering its dual problem. To formulate it, suppose  that $f(z)$ can be
represented in the form $f(z) = F(z,0)$ for all
$z\in\mathcal{Z}$, where $F(z,a)\colon \mathcal{Z}\times\mathcal{A}\to \R$
is some function, and $\mathcal{A}$ is another topological vector space (a convenient
choice of $F$ and $\mathcal{A}$ plays an important role).

Let $\mathcal{A}^*$ denote the topological dual of $\mathcal{A}$. Define the Lagrangian $L\colon \mathcal{Z}\times\mathcal{A}^*\to \Re$ and the
dual objective function $g\colon \mathcal{A}^* \to \Re$ by
\[
L(z,u) = \inf_{a\in \mathcal{A}} \{F(z,a) + \sp au\}, \qquad
g(u) = \inf_{z\in \mathcal{Z}} L(z,u).
\]
Then the dual problem is formulated as the optimization problem
\[
\text{maximize } g(u) \text{ over }u\in \mathcal{A}^*.
\]
If we denote by $V_P$ and $V_D$ the optimal values of the primal and dual
problems respectively (i.e. the infimum of $f(z)$ and the supremum of $g(u)$
respectively), then it is easy too see that $V_P \ge V_D$ always.
 
We are generally interested in the case when the strong duality takes
place, i.e. $V_P=V_D$, or, explicitly,
\begin{equation}
\min_{z\in\mathcal{Z}} f(z) = \max_{u\in\mathcal{A}^*} g(u).
\label{duality}
\end{equation}

Introduce the optimal value function $\phi(a) =
\inf\limits_{z\in\mathcal{Z}} F(z,a)$. The following theorem provides a
sufficient condition for the strong duality \eqref{duality} (see Theorem 7 in \cite{Rockafellar74}).
\begin{theorem} 
\label{th duality}
Suppose $F$ is convex in $(z,a)$
and $\phi(0) = \liminf\limits_{a\to 0} \phi(a)$. Then \eqref{duality} holds.
\end{theorem}

Let us consider a particular case of problem \eqref{opt problem} which
includes constraints in the form of equalities and inequalities. Assume
that $\mathcal{Z} = \Lp$ for some $p\in[1,\infty)$ and two functions
$h_i\colon \Lp\to \L^{r_i}(\R^{n_i})$, $i=1,2$ are given (the spaces
$\L^{p}$ and $\L^{r_i}$ are not necessarily defined on the same probability
space). Consider the problem
\[
\begin{aligned}
\text{minimize}\quad& f(z) \text{ over } z\in \L^p\\
\text{subject to}\quad& g(z) \le 0\text{ a.s.}\\
&h(z) = 0 \text{ a.s.}
\end{aligned}
\]
This problem can be formulated as a particular case of the above abstract
setting by defining
\[
F(z,a_1,a_2) =
\begin{cases}
f(z), &\text{if }g(z) \le a_1\text{ and }h(z) = a_2\text{ a.s.},\\
+\infty, &\text{otherwise}.
\end{cases}
\]
The Lagrangian of this problem is
\begin{align*}
L(z,u_1,u_2) &= \inf_{a_1,a_2} \bigr\{F(z,a_1,a_2) + \sp{a_1}{u_1} + \sp{a_2}{u_2}\bigr\} \\ &=
\begin{cases}
f(z) + \sp {g(z)}{u_1} + \sp {h(z)}{u_2}, &\text{if } u_1\ge0\text{ a.s.},\\
-\infty, &\text{otherwise},
\end{cases}
\end{align*}
where we denote $\sp au = \E(\sum_i a_i u_i)$.

So the dual objective function
\[
g(u,v) = \inf_{z \in L^p} \{f(z) + \sp {g(z)}u + \sp {h(z)}v\}
\qquad\text{for } u\ge0\text{ a.s.},\\
\]
and the dual optimization problem
\begin{align*}
\text{maximize}\quad& g(u,v) \text{ over }u \in L^{r'},\; v\in \L^{w'}\\
\text{subject to}\quad& u\ge 0.
\end{align*}
The strong duality equality:
\[
\min_{z} \{f(z) \mid g(z) \le 0,\; h(z) = 0\} = \max_{u,v} \{g(u,v) \mid u\ge 0\}
\]

\subsection{The minimax theorem}
\begin{theorem}[Sion's minimax theorem, Corollary 3.3 in \cite{Sion58}] 
\label{sion theorem}
Suppose $X,Y$ are convex spaces such that one of them is compact, and $f(x,y)$
is a function on $X\times Y$, such that $x\mapsto f(x,y)$ is quasi-concave
and u.s.c. for each fixed $y$ and $y\mapsto f(x,y)$ is quasi-convex
and l.s.c. for each fixed $x$. Then 
\[
\adjustlimits\sup_{x\in X}\inf_{y\in Y} f(x,y) =  \adjustlimits\inf_{y\in Y} \sup_{x\in X} f(x,y).
\]

\end{theorem}

\end{document}